\documentclass[letterpaper,conference]{ieeeconf}
\IEEEoverridecommandlockouts
\usepackage{mathptmx}
\usepackage{amsthm,amsmath,amssymb,mathtools}
\usepackage{graphicx}
\usepackage{textcomp}
\usepackage{tabularx}
\usepackage{subfigure}
\usepackage{xcolor}
\usepackage{cite}
\usepackage{multirow}
    
\graphicspath{{./images/}}
    
\newtheorem{assumption}{Assumption}
\newtheorem{proposition}{Proposition}
\newtheorem{lemma}{Lemma}

\usepackage{url} 
\usepackage{algorithm,algpseudocode}

\def\BibTeX{{\rm B\kern-.05em{\sc i\kern-.025em b}\kern-.08em
    T\kern-.1667em\lower.7ex\hbox{E}\kern-.125emX}}

\begin{document}

\title{\huge{Security Investment Over Networks with Bounded Rational Agents: Analysis and Distributed Algorithm}
}

\author{Jason Hughes$^1$\thanks{$^1$is with the United States Military Academy Robotics Research Center, West Point, NY 10996 USA. E-Mail: jason.hughes@westpoint.edu} and Juntao Chen$^2$\thanks{$^2$is with the Fordham University Department of Computer and Information Sciences, New York, NY 10023 USA. E-mail: jchen504@fordham.edu}
}

\maketitle

\begin{abstract}
This paper considers the security investment problem over a network in which the resource owners aim to allocate their constrained security resources to heterogeneous targets strategically. Investing in each target makes it less vulnerable, and thus lowering its probability of a successful attack. However, humans tend to perceive such probabilities inaccurately yielding bounded rational behaviors; a phenomenon frequently observed in their decision-making when facing uncertainties. We capture this human nature through the lens of cumulative prospect theory and establish a behavioral resource allocation framework to account for the human's misperception in security investment. We analyze how this misperception behavior affects the resource allocation plan by comparing it with the accurate perception counterpart. The network can become highly complex with a large number of participating agents. To this end,  we further develop a fully distributed algorithm to compute the behavioral security investment strategy efficiently. Finally, we corroborate our results and illustrate the impacts of human's bounded rationality on the resource allocation scheme using cases studies.
\end{abstract}


\section{Introduction}


An efficient allocation of limited security resources to protect targeted assets from malicious attacks is a critical problem faced by security professionals. This problem becomes increasingly challenging as the modern systems adopted in our society become more complex. For example, the societal cyber-physical systems, such as industrial control systems and power grids, consist of heterogeneous components including sensors, controllers, and actuators, which are required to be jointly secured to achieve a desired performance. Thus, it is important for the system operator to allocate the available security resources strategically to enhance the holistic security. Previous studies have investigated how a centralized system operator can maximally reduce the vulnerability of assets from adversaries through security investment \cite{yue2007network,su2017security}. However, such a centralized paradigm, i.e., using a single-source model, is insufficient to capture the emerging scenarios where multiple resource owners\footnote{The resource owner also refers to the system operator/planner, and they are used interchangeably in the paper.} participate in securing the targets collaboratively. To this end, this paper aims to develop a framework and investigate the security investment problem over a network with multiple sources (security resource investors) and a variety of targets (valuable assets to be protected). 

To develop an effective security investment scheme, it is necessary to understand how the risks of targets change over the investment strategy. 
A larger security investment amount will lower the probability of a successful attack. However, human's perception of such probabilistic events is subjective, a commonly observed behavior when facing uncertainties. Psychological studies have shown that humans often misperceive probabilities on gains and losses by over-weighting low probabilities and under-weighting high probabilities which lead to bounded rational behavior, a subject receiving significant attention in prospect theory \cite{Kahneman1979PT,tversky1992cpt}. Such behavioral misperception plays an essential role in the focused security investment problem where the resource owners need to evaluate the likelihood of successful compromise of the targets under a given resource allocation scheme. To this end, we incorporate this bounded rational consideration into our model by developing a new behavioral decision-making framework for security investment over networks. Under this paradigm, the resource owner perceives the target having a relatively low chance of being compromised as more vulnerable than it is, and a target with a high probability of being attacked as less vulnerable than it is. We analyze the impact of attack success misperception on the optimal resource allocation strategy and identify that the bounded rational operator will prefer more to secure those targets with higher values. In other words, the investors tend to pay more attention to higher-valued assets as they become more behavioral, yielding a discriminative distribution scheme comparing with the one without behavioral consideration.

Additionally, when more and more participating agents/nodes (resource owners and targets) are introduced into the network, the computational complexity of the resource allocation problem increases drastically \cite{zhang2019consensus}. Solving the security investment problem with a massive number of networked sources and targets in a centralized manner may be impractical or extremely computationally expensive. In addition, the centralized approach requires the planner to have complete information on the source and target agents, including their utility parameters, supply and demand upper bounds, degree of misperception on attack success, and value of targets. Thus, it does not preserve a high level of privacy for participants. To this end, we propose a distributed algorithm based on the alternating direction method of multipliers (ADMM) \cite{boyd2011distributed}, where the central resource transport planner is not needed and each node solves its own simpler optimization problem and communicates its decisions with the connected nodes in the network. Each pair of source and target nodes will then negotiate to reach a consensus on how many security resources should be transported. The proposed distributed algorithm converges to the same optimal solution obtained under the centralized optimization paradigm. 

The contributions of this paper are summarized as follows:
\begin{enumerate}
    \item We develop a bounded rational security investment framework over a network. The model captures the decision-makers' misperception of security resources' effectiveness on protecting the targets, and it facilitates the analysis of behavioral impacts on the resource allocation plan.
    \item We discover a sequential water-filling nature of the optimal security resource allocation over the targets and identify that the transport planners become more discriminative by investing in a smaller set of higher-valued targets as they tend to be more behavioral.
    \item We further develop a distributed algorithm based on ADMM to compute the optimal resource investment strategy for large-scale networks. We also corroborate our algorithm and analytical results using case studies.
\end{enumerate}

\textit{Related Works:} Optimal security investment for defending assets has been studied extensively in literature \cite{azaiez2007optimalRA,khalili2020resource,huang2019adaptive,xing2021security}. Previous studies have also considered the behavioral impacts on security investment. For example, the authors in \cite{yang2013improvingRA} have studied the interplay between the strategic defender and the bounded rational adversary through a game-theoretic framework. \cite{mustafa2019protecting} has investigated the optimal investment strategies under a misperceived security risk model based on prospect theory, in which the authors have focused on a single decision maker investing on heterogeneous assets. Our work pays attention to the resource transport in a multi-source multi-target framework. 
Prospect theory has also been used to guide the optimal resource allocation/decision-making in various applications, including water infrastructures \cite{he2019water}, communication networks \cite{vamvakas2019exploiting}, and the Internet of things \cite{tian2021honeypot}. 
Our work is also related to the decision-making of resource allocation over large-scale networks, in which efficient algorithms for computing the optimal schemes have been proposed in different contexts, including in consideration of efficiency \cite{zhang2019consensus,xiong2018optimal}, fairness \cite{jhughes2021fair}, and security and resiliency \cite{hughes2021deceptive,farooq2018adaptive}.

The rest of this paper is organized as follows. In Section II, we formulate the security investment problem over a network that considers human misperception on attack success rate. We characterize the optimal security investment strategy and analyze how the behavioral consideration affects such solutions in Sections III and IV. We further develop a distributed algorithm to compute the behavioral security investment strategy in Section V and corroborate our findings in Section VI.

\section{Problem Formulation}
In this section, we establish a framework for security investment over a network with behavioral participants.

\subsection{Security Resource Transport Network}
In a network, we denote by $\mathcal{X}:=\{1, ..., |\mathcal{X}|\}$ the set of destinations/targets that receive the security resources, and $\mathcal{Y}:=\{1, ..., |\mathcal{Y}|\}$ the set of origins/sources that distribute security resources to the targets. Specifically, each source node $y\in\mathcal{Y}$ is connected to a number of target nodes denoted by $\mathcal{X}_y$, representing that $y$ has choices in allocating its resources to a specific group of destinations $\mathcal{X}_y$ in the network. Similarly, it is possible that each target node $x\in\mathcal{X}$ receives resources from multiple source nodes, and this set of suppliers to node $x$ is denoted by $\mathcal{Y}_x$. Note that $\mathcal{X}_y$, $\forall y$ and $\mathcal{Y}_x$, $\forall x$ are nonempty. It is straightforward to see that the security resources are transported over a bipartite network, where one side of the network consists of all source nodes and the other includes all destination nodes. This bipartite graph may not be complete due to constrained matching policies between participants. For convenience, we denote by $\mathcal{E}$ the set including all feasible transport paths in the network, i.e., $\mathcal{E}:=\{(x,y)|x\in\mathcal{X}_y,y\in\mathcal{Y}\}$. Note that $\mathcal{E}$ also refers to the set of all edges in the established bipartite graph for security resource transportation.

We next denote by $\pi_{xy}\in\mathbb{R}_+$ the amount of security resources transported from the origin node $y\in\mathcal{Y}$ to the destination node $x\in\mathcal{X}$, where $\mathbb{R}_+$ is the set of nonnegative real numbers. Let $\Pi:=\{\pi_{xy}\}_{x\in\mathcal{X}_y,y\in\mathcal{Y}}$ be the designed resource transport plan. Furthermore, the security resources at each source node $y\in\mathcal{Y}$ is upper bounded by $\bar{q}_{y}\in\mathbb{R}_+$, i.e., $\sum_{x\in\mathcal{X}_y}\pi_{xy}\leq \bar{q}_{y}$.

\subsection{Bounded Rational Security Investment} 
Each target node in the network faces threats and could be compromised by an attacker. If target node $x\in\mathcal{X}$ is attacked, the induced loss is $U_x>0$. The attacker's probability of successfully compromising the target node is related to the amount of security resources received. For each target node $x\in\mathcal{X}$, defined by $p_x:\mathbb{R}_+^{|\mathcal{Y}_x|}\rightarrow [0,1]$ a function that maps the received security resources $\Pi_x$ to a successful attack probability. It is natural to see that such probability should be related to the aggregated resource received by node $x$ captured by $\sum_{y\in\mathcal{Y}_x}\pi_{xy}$. Thus, with a slightly abuse of notation, $p_x(\Pi_x)$ can be expressed by $p_x(\sum_{y\in\mathcal{Y}_x}\pi_{xy})$, where the later one shows more explicitly the relationship between the successful attack probability and the total received resources at node $x\in\mathcal{X}$. 

Each target node $x\in\mathcal{X}$ minimizes its cost $U_x p_x(\Pi_x)$. To this end, the central planner aims to minimize the following aggregated loss $L(\Pi)$ at all targets under attacks:
\begin{equation} \label{eqn:loss_func}
    L(\Pi) = \sum_{x\in\mathcal{X}} U_x p_x(\Pi_x).
\end{equation}

It has been shown that humans tend to misperceive probabilities by over-weighing low probabilities and under-weighing high probabilities during decision-making under uncertainties. For a  true probability $p\in[0,1]$, humans will perceive it as $w(p)\in[0,1]$, where $w$ is a probability weighting function. One such commonly used weighting function is given by \cite{prelec1998}
\begin{equation} \label{eqn:weighting_func}
    w(p) = \exp({-(-\log(p))^\gamma}),\ p\in[0,1],
\end{equation}
where $\gamma\in(0,1]$ is a parameter capturing the degree of misperception. When $\gamma$ is closer to 0, it leads to a larger distortion of the probability function $p$. In comparison, when $\gamma=1$, $w(p) = p$, indicating there is no probability misperception.

Under the perceived probability, the target node $x$'s cost function becomes $U_x w\left(p_x(\Pi_x)\right)$. Thus, the cost function for the transport planner under the perceived attack probability is
\begin{equation} \label{eqn:weighted_loss}
    \tilde{L}(\Pi) = \sum_{x\in\mathcal{X}} U_x w(p_x(\Pi_x)).
\end{equation}

The security resource allocation strategies with bounded rational behavioral consideration can be obtained by solving the following optimization problem:
\begin{equation}\label{eqn:weighted}
\begin{aligned}
\mathrm{(OP-A):}\quad \min_{\Pi}\ &\sum_{x\in\mathcal{X}} U_x w(p_x(\Pi_x))\\
    \mathrm{s.t.}\ 
    &0\leq \sum_{x\in\mathcal{X}_y} \pi_{xy}\leq \bar{q}_{y},\ \forall y\in\mathcal{Y},\\
    &\pi_{xy}\geq 0,\ \forall \{x,y\} \in\mathcal{E}.
\end{aligned}
\end{equation}

Note that (OP-A) is solved for a distribution of resources across the source nodes at a given time. If the amount of resources at each node changes, a new optimal strategy can be obtained by solving (OP-A) repeatedly in a moving horizon fashion. Extending (OP-A) to dynamic resource allocation over a period a time is possible and left as subsequent work.

\section{Preliminary Analysis}

The successful attack probability function $p_x$ should capture the fact that a larger security investment lowers the likelihood of attack. In addition, the marginal benefit of security resource decreases for each target node. 
To this end, we have the following assumption.
\begin{assumption}\label{assump:1}
The successful attack probability function $p_x(\Pi_x)$ satisfies the following: 1) $p_x(\Pi_x)\in[0,1]$ with $\lim_{||\zeta||_1\rightarrow \infty} p_x(\zeta)$ = 0, where $||\cdot||_1$ denotes the $l_1$ norm, and is twice differentiable, 2) $p_x(\Pi_x)$ is strictly monotonic decreasing and log-convex with respect to $\pi_{xy}$, for $y\in\mathcal{Y}_x$, and 3) $\frac{\partial p_x}{\partial \pi_{xy}}/p_x$ is bounded with respect to $\pi_{xy}$, for $y\in\mathcal{Y}_x$.
\end{assumption}

There are a number of functions of interest that satisfy the
properties in Assumption \ref{assump:1}. For example,
\begin{equation} \label{eqn:prelec}
    p_x(\Pi_x) = \exp({-\sum_{y\in\mathcal{Y}_x}\pi_{xy} - r_x}),
\end{equation}
where $r_x>0$ represents the existing security investment at node $x$ before resource transport.

As another example, $p_x(\Pi_x) = \frac{1}{\sum_{y\in\mathcal{Y}_x}\pi_{xy} + r_x}$, where $r_x>1$ has a similar meaning as in the previous case. Both examples indicate that the security resources can effectively decrease the attack likelihood.

\begin{lemma}\label{lemma:convex}
Under Assumption \ref{assump:1}, the perceived probability of successful attack at node $x$, $w(p_x(\Pi_x))$, is strictly convex in $\pi_{xy}$, $\forall x\in\mathcal{X},\ y\in\mathcal{Y}_x$.
\end{lemma}
\begin{proof}
See Appendix \ref{lemma_proof:convex}.
\end{proof}

\section{Analysis of Bounded Rational Security Investment Strategies}

This section characterizes the bounded rational security investment strategies and analyzes the impacts of behavioral considerations on such decision-making outcomes. 

\subsection{Security Resource Allocation Preferences}
We first have the following assumption to facilitates the analysis.
\begin{assumption}\label{assump:2}
Assume that the values of induced loss due to successful attack are ordered as follows: $U_1>U_2>...>U_{|\mathcal{X}|}>0$. Furthermore, each target node admits a same successful attack probability function, i.e., $p_x$, $\forall x\in\mathcal{X}$, share a same form.
\end{assumption}

We next have the following result on the marginal cost associated with the target nodes.

\begin{lemma}\label{lemma:2}
The following inequality holds for each pair of target nodes $i\in\mathcal{X}$ and $j\in\mathcal{X}$ with $i<j$,
\begin{equation}\label{eqn:marginal}
U_{i} \frac{\partial w(p_i(\Pi_i))}{\partial \pi_{iy}}<U_j\frac{\partial w(p_j(\Pi_j))}{\partial \pi_{jy}},\ \forall y\in\mathcal{Y}_i\cap \mathcal{Y}_j.
\end{equation}
And the marginal cost $U_{i} \frac{\partial w(p_i(\Pi_i))}{\partial \pi_{iy}}$ is negative and continuously increasing to 0 in $\pi_{iy}$, $\forall i\in\mathcal{X}$.
\end{lemma}
\begin{proof}
First, we have, $\forall x\in\mathcal{X}$ and $\forall y\in\mathcal{Y}_x$,
\begin{equation}\label{eqn:marginal_def}
\begin{aligned}
U_{x} \frac{\partial w(p_x(\Pi_x))}{\partial \pi_{xy}} =&U_{x}\gamma(-\log(p_x(\Pi_x)))^{(\gamma-1)}\\ & \cdot \frac{\partial p_x(\Pi_x)}{\partial \pi_{xy}}\Big/p_x(\Pi_x) \cdot w(p_x(\Pi_x)).
\end{aligned}
\end{equation}
Based on Assumption \ref{assump:1}, $\frac{\partial p_x(\Pi_x)}{\partial \pi_{xy}}$ is negative. In addition, $-\log(p_x(\Pi_x))$, $p_x(\Pi_x)$ and $w(p_x(\Pi_x))$ are all positive. Thus, $U_{x} \frac{\partial w(p_x(\Pi_x))}{\partial \pi_{xy}}<0$. Lemma \ref{lemma:convex} shows that $\frac{\partial}{\partial \pi_{xy}}(\frac{\partial w(p_x(\Pi_x))}{\partial \pi_{xy}})>0$, indicating that the marginal cost is monotonically increasing. In addition,
$
    \lim_{||\Pi_x||_1 \rightarrow \infty} \left|U_x \frac{\partial w(p(\Pi_x))}{\partial \pi_{xy}} \right| 
    =\lim_{||\Pi_x||_1 \rightarrow \infty} \left|\gamma U_x (-\log(p_x(\Pi_x)))^{\gamma-1} w(p_x(\Pi_x)) \right| \left|\frac{\partial p_x(\Pi_x)}{\partial \pi_{xy}}\big/p_x(\Pi_x) \right|.
$
From Assumption \ref{assump:1}, we know that as $p_x(\Pi_x) = 0$ and thus $w(p_x(\Pi_x)) \rightarrow 0$ and $-\log(p_x(\Pi_x)) \rightarrow \infty$ as $||\Pi_x||_1 \rightarrow \infty$. Since $\frac{\partial p_x}{\partial \pi_{xy}}/p_x$ is bounded, $\lim_{||\Pi_x||_1 \rightarrow \infty} |U_x \cdot \frac{\partial w(p(\Pi_x))}{\partial \pi_{xy}}|=0$. Finally, to show the inequality between the marginals, we note that based on Assumption \ref{assump:2}, $U_{i} > U_{j}$ and thus
\begin{equation*}
\begin{split}
    &U_{i}(-\log(p(\Pi_i)))^{\gamma-1}w(p(\Pi_i))\frac{\partial p_i(\Pi_i)}{\partial \pi_{iy}}\Big/p_i(\Pi_i) \\ &<U_{j}(-\log(p(\Pi_i)))^{\gamma-1}w(p(\Pi_i))\frac{\partial p_j(\Pi_i)}{\partial \pi_{jy}}\Big/p_j(\Pi_i),
\end{split}
\end{equation*}
$\forall \Pi_x \in \mathbb{R}_+^{|\mathcal{Y}_x|}, \forall x \in \mathcal{X}$ which yields \eqref{eqn:marginal} because $\frac{\partial p_j(\Pi_i)}{\partial \pi_{jy}}$ is negative.
\end{proof}

The following proposition characterizes the total amount of security resources received by the target nodes from sources under some general assumptions.
\begin{proposition}\label{prop:1}
Under Assumption \ref{assump:2}
and a complete resource transport network, the optimal strategy $\{\Pi_x^*\}_{x\in\mathcal{X}}$ satisfies the following inequality $||\Pi_1^*||_1\geq ||\Pi_2^*||_1\geq...\geq ||\Pi_{|\mathcal{X}|}^*||_1$, and the $l_1$ norm $||\Pi_x||_1=\sum_{y\in\mathcal{Y}_x}\pi_{xy}$ is the total amount of security resources received by the target node $x\in\mathcal{X}$.
\end{proposition}
\begin{proof}
As each source can transfer resources to every target and the qualities of resources are the same supplied by all source nodes, it is equivalent to aggregate all the source nodes as a single super node that has a capacity $\sum_{y\in\mathcal{Y}}\bar{q}_y$ managed by a central planner. Thus, the transport network can be seen consisting of a single source connected to a set of targets, i.e., $\mathcal{Y}=\{1\}$, $\mathcal{Y}_x=\mathcal{Y}$, and $\Pi_x=\pi_{x1}$, $\forall x\in\mathcal{X}$. Based on the KKT condition, for each pair of target nodes $i\in\mathcal{X}$ and $j\in\mathcal{X}$ receiving nonzero security resources from the source node, we have $U_i \frac{\partial w(p_i(\Pi_i))}{\partial \pi_{i1}}|_{\pi_{i1}=\pi_{i1}^*}=U_j \frac{\partial w(p_i(\Pi_j))}{\partial \pi_{j1}}|_{\pi_{j1}=\pi_{j1}^*}$. Assumptions \ref{assump:1} and \ref{assump:2} indicate that $\gamma U_i(-\log(p_i(\pi_{i1}^*)))^{\gamma-1}w(p_i(\pi_{i1}^*))\frac{\partial p_i(\pi_{i1}^*)}{\partial \pi_{i1}}/p_i(\pi_{i1}^*) = \gamma U_j(-\log(p_j(\pi_{j1}^*)))^{\gamma-1}w(p_j(\pi_{j1}^*))\frac{\partial p_j(\pi_{j1}^*)}{\partial \pi_{j1}}/p_j(\pi_{j1}^*)$, which yields
\begin{align*}
    &(-\log(p_i(\pi_{i1}^*)))^{\gamma-1}w(p_i(\pi_{i1}^*))\frac{\partial p_i(\pi_{i1}^*)}{\partial \pi_{i1}}\frac{1}{p_i(\pi_{i1}^*)} \\
    &= \frac{U_j}{U_i}(-\log(p_j(\pi_{j1}^*)))^{\gamma-1}w(p_j(\pi_{j1}^*))\frac{\partial p_j(\pi_{j1}^*)}{\partial \pi_{j1}}\frac{1}{p_j(\pi_{j1}^*)}\\
    & > (-\log(p_j(\pi_{j1}^*)))^{\gamma-1}w(p_j(\pi_{j1}^*))\frac{\partial p_j(\pi_{j1}^*)}{\partial \pi_{j1}}\frac{1}{p_j(\pi_{j1}^*)}.
\end{align*}
The inequality in the last step above is due to $U_j<U_i$ for $j>i$ and the negative marginal cost of target node on the received security resources. Based on Lemma \ref{lemma:2}, the marginal cost is continuously increasing. Thus, we can conclude $\pi_{i1}^*> \pi_{j1}^*$, where $\pi_{i1}^*$ is the total amount of resources received by target node $i$. Equivalently speaking, target node $i$ receives more security resources than target node $j$ at the optimal solution, for $i<j\in\mathcal{X}$.
\end{proof}

\textit{\textbf{Remark:}} Proposition \ref{prop:1} indicates that, under some quite general conditions, target node $i$ (higher-valued) receives more resources than target node $j$ (lower-valued) under the optimal allocation plan, $\forall i<j\in\mathcal{X}$. This result is consistent with the objective of the system planner in minimizing the aggregated expected loss of assets.

\subsection{Sequential Water-Filling of Security Investment}
The transport network is still considered to be complete. Thus, it is equivalent to combine all source nodes and regard them as a super source node with capacity $\sum_{y\in\mathcal{Y}}\bar{q}_y$. For convenience, we denote by $\tilde{\pi}_i$ the total amount of resources that target node $i$ received from the super source node, i.e., $\tilde{\pi}_i = \sum_{y\in\mathcal{Y}_i}\pi_{iy}$ in the original framework. With a slight abuse of notation, $p_x$ can be seen as a single-variable function on $\tilde{\pi}_i$.
For all target nodes $i\in\mathcal{X}$ and $j\in\mathcal{X}$ with $i<j$, we define $\tilde{\pi}_i^{j*}$ as a quantity that satisfies
\begin{equation}\label{eqn:threshold}
U_{i} \frac{\partial w(p_i(\tilde{\pi}_i))}{\partial \tilde{\pi}_i}\Bigg\vert_{\tilde{\pi}_i=\tilde{\pi}_i^{j*}}=U_j\frac{\partial w(p_j(\tilde{\pi}_j))}{\partial \tilde{\pi}_j}\Bigg\vert_{\tilde{\pi}_j=0}.
\end{equation}

\begin{proposition}\label{prop:2}
Under Assumption \ref{assump:2} and a complete transport network, the resources received by target node $j$ from the super source node, $\tilde{\pi}_j$, will be nonzero at the optimal solution if and only if  $\sum_{y\in\mathcal{Y}}\bar{q}_y>\sum_{i=1}^{j-1} \tilde{\pi}_i^{j*}$, where $\tilde{\pi}_i^{j*}$ is defined in \eqref{eqn:threshold}.
\end{proposition}

\begin{proof}
Suppose that  $\tilde{\pi}_{j}^*>0$ for some target node $j$. Also suppose by contradiction that $\sum_{y\in\mathcal{Y}} \bar{q}_y \leq \sum_{i=1}^{j-1}\tilde{\pi}_{i}^{j*}$. Then, $\exists m \in \{1,...,j-1\}$ such that $ \tilde{\pi}_m^* < \tilde{\pi}_{m}^{j*}$. This indicates that it is infeasible to allocate $\tilde{\pi}_{m}^{j*}$ or more resources to node $m$ without exceeding the upper bound.
By definition of $\tilde{\pi}_{m}^{j*}$,
\begin{equation*}
\begin{aligned}
    U_m \frac{\partial w(p_{m}(\tilde{\pi}_{m}))}{\partial \tilde{\pi}_{m}}\bigg\rvert_{\tilde{\pi}_{m}= \tilde{\pi}_m^{*}} < U_{j}\frac{\partial w(p_{j}(\tilde{\pi}_{j}))}{\partial \tilde{\pi}_{j}}\bigg\rvert_{\tilde{\pi}_{j}= 0} \\< U_{j}\frac{\partial w(p_{j}(\tilde{\pi}_{j}))}{\partial \tilde{\pi}_{j}}\bigg\rvert_{\tilde{\pi}_{j}= \tilde{\pi}_{j}^*},
\end{aligned}
\end{equation*}
which yields a contradiction, since the marginals must coincide at the optimal solution.
Thus, $ \tilde{\pi}_j^* > 0$ leads to $\sum_{y \in \mathcal{Y}}\bar{q}_y>\sum_{i=1}^{j-1}\tilde{\pi}_i^{j*}$ under the optimal resource allocation.

To prove the other direction, we first suppose that $\sum_{y \in \mathcal{Y}}\bar{q}_y>\sum_{i=1}^{j-1}\tilde{\pi}_{i}^*$ and suppose by contradiction $\tilde{\pi}_j = 0$. Then we have $\tilde{\pi}_k = 0, \forall k>j$, and thus $\sum_{k=1,2,...,j-1}\tilde{\pi}_k = \sum_{y \in \mathcal{Y}}\bar{q}_y$ and $\exists i \in \{1,...,j-1\}$ such that $\tilde{\pi}_i > \tilde{\pi}_i^{j*}$.
A sufficiently small amount of resource, $\epsilon \in \mathbb{R}_+$ is transferred from target $i$ to $j$ which will lead to a net cost reduction in (OP-A), and thus the resource allocation is no longer optimal.
Starting with non-zero resource allocation to the target nodes $\{1,...,j-1\}$, the total cost is
$
\sum_{x \in \mathcal{X}} U_x w(p_x(\tilde{\pi}_x)).
$
From target $i$ that has $\tilde{\pi}_i^* = ||\Pi_i^*||_1 >  \tilde{\pi}_i^{j*}$, remove a sufficiently small amount of resource $\epsilon$ and add a resource amount of $\epsilon$ to target $j$. Denote the modified resource transport plan as $\pi^{(\epsilon)}$. The total cost after perturbation becomes
\begin{align*}
    \tilde{L}(\pi^{(\epsilon)}) 
    = \sum_{z \in \mathcal{X}\setminus\{i,j\}} U_zw (p_z(\tilde{\pi}_z^*)) + U_i w(p_i(\tilde{\pi}_i - \epsilon)) \\+ U_j w(p_j(\epsilon)).
\end{align*}
We next define $g(\epsilon) = U_i w(p_i(\tilde{\pi}_i - \epsilon)) + U_j w(p_j(\epsilon))$. Then,
$
\tilde{L}(\pi^{*}) = \sum_{z \in \mathcal{X}\setminus\{i,j\}} U_z w(p_z(\tilde{\pi}_z)) + g(0),\ 
\tilde{L}(\pi^{(\epsilon)}) = \sum_{z \in \mathcal{X}\setminus\{i,j\}} U_z w(p_z(\tilde{\pi}_z)) + g(\epsilon).
$
If $g(\epsilon) < g(0)$, then $\tilde{L}(\pi^{(\epsilon)})<\tilde{L}(\tilde{\pi}^*)$, which yields a positive net cost reduction, meaning that the resource allocation strategy after perturbation is worse off. It is clear that
$$
    \frac{dg}{d\epsilon} 
    = -U_i \frac{\partial w(p_i(\pi_i))}{\partial \tilde{\pi}_{i}}\bigg\rvert_{\pi_i = \tilde{\pi}_{i} - \epsilon} + U_j \frac{\partial w(p_j(\pi_j))}{\partial \tilde{\pi}_{j}}\bigg\rvert_{\pi_j= \epsilon}.
$$
Based on $\tilde{\pi}_i^* >\tilde{\pi}_i^{j*}$ and Lemma \ref{lemma:2},
$$
U_i \frac{\partial w(p_i(\tilde{\pi}_i))}{\partial \tilde{\pi}_{i}}\bigg\rvert_{\tilde{\pi}_i = \tilde{\pi}_i^* } > U_j \frac{\partial w(p_j(\tilde{\pi}_j))}{\partial \tilde{\pi}_{j}}\bigg\rvert_{\tilde{\pi}_j= 0}.
$$
Thus, $\lim_{\epsilon \rightarrow 0} \frac{dg}{d\epsilon}$ is negative, indicating that  $g(\epsilon)$ is decreasing for a sufficiently small $\epsilon$. Therefore, we obtain $\tilde{L}(\pi^{(\epsilon)}) < \tilde{L}(\tilde{\pi}^*)$ which is a contradiction.
\end{proof}

\textit{\textbf{Remark:}} Proposition \ref{prop:2} implies that the super source node first allocates $\tilde{\pi}_1^{2*}$ security resources to target node 1, and then starts to transfer resources to both target nodes 1 and 2 while maintaining a same marginal cost until reaching $\tilde{\pi}_1^{3*}$ and $\tilde{\pi}_2^{3*}$, respectively. Afterward, in addition to target nodes 1 and 2, target node 3 starts to receive resources, and the marginal costs at all nodes are kept the same during security resource investment. The resource allocation scheme will follow this fashion until all resources are transferred. The above discussion leads to \textit{sequential water-filling} of security resource transport over networks.

As the original transport network includes multiple source nodes, we need to determine the strategy for each of them. The above discussion indicates that the optimal resource allocation plan can be obtained sequentially, i.e., each source node completes allocating its security resources to the targets in sequential order. Specifically, source node 1 will first transfer its resource to target node 1. If $\bar{q}_1<\tilde{\pi}_1^{2*}$, then the next source nodes (node 2, 3, etc) will continue allocate resource to target node 1 until it receives $\tilde{\pi}_1^{2*}$ amount of resources. If $\bar{q}_1>\tilde{\pi}_1^{2*}$, then source node 1 first allocates $\tilde{\pi}_1^{2*}$ amount of resources to the target node 1, and then start transferring the remaining resources to both targets 1 and 2 while maintaining a same marginal cost at both nodes. After source node 1 completes its resource transport, source node 2 starts to transfer its resources to the appropriate targets in a similar manner. This process terminates until all source nodes finish their resource allocation to the targets. 

\subsection{Behavioral Impacts on Security Resource Allocation}
The impact of incorporating the behavioral element to probability perception is captured by the parameter $\gamma$. Clearly, there is no behavioral consideration when $\gamma = 1$, and the probabilities are perceived as their actual values. When $\gamma\in(0,1)$, we have the following result on the behavioral impacts on the resource allocation plan.

\begin{proposition}
Under Assumptions \ref{assump:1} and \ref{assump:2}, $p_x(0)<\frac{1}{e}$, $\forall x\in\mathcal{X}$, and a complete transport network, ${d \tilde{\pi}_i^{j*}}/{d \gamma}<0$ for $i<j$, $\forall i,j\in\mathcal{X}$, with $\tilde{\pi}_i^{j*}$ defined in \eqref{eqn:threshold}.
\end{proposition}

\begin{proof}
Based on \eqref{eqn:weighting_func} and \eqref{eqn:threshold}, we have
\begin{equation}\label{eqn:threshold2}
\begin{aligned}
    U_i\gamma(-\log(p_i(\tilde{\pi}_i^{j*})))^{(\gamma-1)} w(p_i(\tilde{\pi}_i^{j*})) \frac{\partial p_i(\tilde{\pi}_i)}{\partial \tilde{\pi}_i}\bigg\rvert_{\tilde{\pi}_i=\tilde{\pi}_i^{j*}} \frac{1}{p_i(\tilde{\pi}_i^{j*})} \\
   = U_j\gamma(-\log(p_j(0)))^{(\gamma-1)} w(p_j(0)) \frac{\partial p_j(\tilde{\pi}_j)}{\partial \tilde{\pi}_j}\bigg\rvert_{\tilde{\pi}_j=0} \frac{1}{p_j(0)}.
\end{aligned}
\end{equation}
Based on \eqref{eqn:threshold2}, we can characterize the sensitivity of the amount of security resources transported to each target over the behavioral parameter $\gamma$. Taking log of each side of \eqref{eqn:threshold2} and differentiating with respect to $\gamma$ yield
\begin{equation}\label{derivative_gamma}
\begin{aligned}
        \frac{d \tilde{\pi}_i^{j*}}{d \gamma} = \frac{((-\log(p_i(\tilde{\pi}_i^{j*})))^{\gamma}-1) \log(-\log(p_i(\tilde{\pi}_i^{j*})))}{\Lambda_{i}^{j}} \\
        -\frac{((-\log(p_j(0)))^{\gamma}-1)  \log(-\log(p_j(0)))}{\Lambda_{i}^{j}},
\end{aligned}
\end{equation}
where
$
        \Lambda_{i}^{j} = (\gamma - 1 - \gamma(-\log(p_i(\tilde{\pi}_i^{j*})))^{\gamma}) \frac{\partial p_i(\tilde{\pi}_i)}{\partial \tilde{\pi}_i}\bigg\rvert_{\tilde{\pi}_i=\tilde{\pi}_i^{j*}} \cdot\frac{1}{p_i(\tilde{\pi}_i^{j*})} \\ \log(p_i(\tilde{\pi}_i^{j*}))+
        \frac{p_i(\tilde{\pi}_i^{j*})\frac{\partial^2 p_i(\tilde{\pi}_i)}{\partial \tilde{\pi}_i^2}\big\rvert_{\tilde{\pi}_i=\tilde{\pi}_i^{j*}}-\big(\frac{\partial p_i(\tilde{\pi}_i)}{\partial \tilde{\pi}_i}\big\rvert_{\tilde{\pi}_i=\tilde{\pi}_i^{j*}}\big)^2}{\frac{\partial p_i(\tilde{\pi}_i)}{\partial \tilde{\pi}_i}\big\rvert_{\tilde{\pi}_i=\tilde{\pi}_i^{j*}}\cdot p_i(\tilde{\pi}_i^{j*})}.$
Under the assumption that $p_j(0)<\frac{1}{e}$ and $p_i(\tilde{\pi}_i^{j*}) < p_j(0)$ for $\tilde{\pi}_i^{j*} > 0$, we have $-\log(p_i(\tilde{\pi}_i^{j*})) > -\log(p_j(0))>1$ and thus $\log(-\log(p_i(\tilde{\pi}_i^{j*}))) > \log(-\log(p_j(0))) >0$ and  $(-\log(p_i(\tilde{\pi}_i^{j*})))^{\gamma} - 1 >(-\log(p_j(0)))^{\gamma} - 1$. Hence, the numerator of \eqref{derivative_gamma} is positive. From Assumption \ref{assump:1}, we have $\frac{\partial p_i(\tilde{\pi}_i)}{\partial \tilde{\pi}_i}\bigg\rvert_{\tilde{\pi}_i=\tilde{\pi}_i^{j*}} < 0$ and because $p_i(\tilde{\pi}_i)$ is log-convex, $p_i(\tilde{\pi}_i^{j*})\frac{\partial^2 p_i(\tilde{\pi}_i)}{\partial \tilde{\pi}_i^2}\big\rvert_{\tilde{\pi}_i=\tilde{\pi}_i^{j*}}\geq \big(\frac{\partial p_i(\tilde{\pi}_i)}{\partial \tilde{\pi}_i}\big\rvert_{\tilde{\pi}_i=\tilde{\pi}_i^{j*}}\big)^2$. Thus, the denominator of \eqref{derivative_gamma} negative, which yields ${d \tilde{\pi}_i^{j*}}/{d \gamma}<0$. 
\end{proof}

\textit{\textbf{Remark:}} The above analysis, together with Proposition \ref{prop:1} indicate that when the behavioral misperception on the attack success probability is considered, the sources will supply security resources to fewer target nodes than the optimal strategy obtained under the non-behavioral counterpart. In other words, the behavioral security resource owners prefer to secure higher-valued assets while paying less attention to those relatively lower-valued targets, and it leads to a \textit{discriminative resource allocation scheme} comparing with the one developed under fully rational scenario.

\section{Distributed Algorithm for Bounded Rational Security Investment}
This section aims to develop a distribution computational scheme to obtain the behavioral security investment scheme.

In the established framework, the objective function can also incorporate the preferences of the source nodes in the security resource transport design in addition to the cost of the target nodes. The utility function of source node $y$ on transferring $\pi_{xy}$ security resources to target node $x$ is denoted by $s_{xy}:\mathbb{R}_+\rightarrow\mathbb{R}$.  In addition, to balance the security resource allocation, the planner considers an upper bound of each target node $x\in\mathcal{X}$ on the received security resources from connected sources, captured by $\bar{p}_x\in\mathbb{R}_+$, i.e., $\sum_{y\in\mathcal{Y}_x} \pi_{xy}\leq \bar{p}_{x}$. To this end,
the system planner aims to address:
\begin{equation*}
\begin{aligned}
\mathrm{(OP-B):}\quad \min_{\Pi}\ &\sum_{x\in\mathcal{X}} U_x w(p_x(\Pi_x))-\sum_{y\in\mathcal{Y}} \sum_{x\in\mathcal{X}_y} \tau_y s_{xy}(\pi_{xy})\\
    \mathrm{s.t.}\ 
    &0\leq \sum_{y\in\mathcal{Y}_x} \pi_{xy}\leq \bar{p}_{x},\ \forall x\in\mathcal{X},\\
    &0\leq \sum_{x\in\mathcal{X}_y} \pi_{xy}\leq \bar{q}_{y},\ \forall y\in\mathcal{Y},\\
    &\pi_{xy}\geq 0,\ \forall \{x,y\} \in\mathcal{E},
\end{aligned}
\end{equation*}
where $\tau_y \in \mathbb{R}_+$ is a positive weighting factor balancing the loss of the targets and the utility of the sources under a given security allocation strategy. It is straightforward to observe that as $\tau_y\rightarrow 0,\ \forall y$, the solution to (OP-B) will converge to the one of (OP-A), given that the targets have no constraint on the maximum received security resources.

As the resource transport network becomes complex with a large number of participating nodes, a centralized scheme to compute the optimal solution to (OP-B) can be computationally expensive. In addition, the centralized optimization paradigm requires the planner to collect heterogeneous information from all source and target nodes, including their utility parameters, supply and demand upper bounds, degree of misperception on attack success, and value of targets, which does not preserve a desirable level of privacy. Due to the above two concerns, it is necessary to devise a \textit{distributed} and \textit{privacy-preserving} scheme to obtain the behavioral resource allocation strategy over a large-scale network. 

To facilitate the development of such an algorithm, we first introduce two ancillary variables $\pi_{xy}^t$ and $\pi_{xy}^s$. The superscripts $t$ and $s$ indicate that the corresponding parameter belongs to a target and source node, respectively. We then set $\pi_{xy} = \pi_{xy}^t$ and $\pi_{xy} = \pi_{xy}^s$, indicating that the solutions proposed by the targets and sources are consistent. This reformulation facilitates the design of a distributed algorithm which allows us to iterate to obtain the optimal strategy. To this end, the reformulated problem is presented as follows:
\begin{equation}\label{eqn:ancillary_intro}
\begin{aligned}
\min_{\Pi^t \in \mathcal{F}^t, \Pi^s \in \mathcal{F}^s,\Pi} & \sum_{x\in\mathcal{X}}  U_x w(p_x(\Pi_x^t)) - \sum_{y\in\mathcal{Y}} \sum_{x\in\mathcal{X}_y} \tau_y s_{xy}(\pi_{xy}^s) \\
\mathrm{s.t.}\quad & \pi_{xy}^t = \pi_{xy},\ \forall \{x,y\}\in\mathcal{E},\\
& \pi_{xy}^s = \pi_{xy},\ \forall \{x,y\}\in\mathcal{E},
\end{aligned}
\end{equation}
where $\Pi^t:=\{\pi_{xy}^t\}_{x\in\mathcal{X}_y,y\in\mathcal{Y}}$, $\Pi^s:=\{\pi_{xy}^s\}_{x\in\mathcal{X},y\in\mathcal{Y}_x}$, $\mathcal{F}^t := \{ \Pi^t | \pi_{xy}^t \geq 0, \underline{p}_x \leq \sum_{y \in \mathcal{Y}_x} \pi_{xy}^t \leq \bar{p}_x,\ \{x,y\} \in \mathcal{E}\}$, and $\mathcal{F}^s := \{ \Pi^s | \pi_{xy}^s \geq 0, \underline{q}_y \leq \sum_{x \in \mathcal{X}_y} \pi_{xy}^s \leq \bar{q}_y,\ \{x,y\} \in \mathcal{E} \}$.

We resort to alternating direction method of multipliers (ADMM) \cite{boyd2011distributed} to develop a distributed computational algorithm. First, let $\alpha_{xy}^s$ and $\alpha_{xy}^t$ be the Lagrangian multipliers associated with the constraint $\pi_{xy}^s = \pi_{xy}$ and $\pi_{xy}^t = \pi_{xy}$, respectively. The Lagrangian function associated with the optimization problem \eqref{eqn:ancillary_intro} can then be written as follows:
\begin{equation}\label{eqn:lagrangian}
\begin{aligned}
\mathcal{L}&(\Pi^t,\Pi^s,\Pi,\alpha_{xy}^t,\alpha_{xy}^s) \\
&= \sum_{x\in\mathcal{X}} U_x w(p_x(\Pi_x^t)) - \sum_{y\in\mathcal{Y}} \sum_{x\in\mathcal{X}_y} \tau_y s_{xy}(\pi_{xy}^s) \\
&+ \sum_{x\in\mathcal{X}} \sum_{y\in\mathcal{Y}_x} \alpha_{xy}^t \left( \pi_{xy}^t - \pi_{xy} \right)  
+ \sum_{y\in\mathcal{Y}} \sum_{x\in\mathcal{X}_y} \alpha_{xy}^s \left( \pi_{xy} - \pi_{xy}^s \right) \\ &+\frac{\eta}{2} \sum_{x\in\mathcal{X}} \sum_{y\in\mathcal{Y}_x} \left( \pi_{xy}^t - \pi_{xy} \right)^2 
+\frac{\eta}{2} \sum_{x\in\mathcal{X}} \sum_{y\in\mathcal{Y}_x} \left( \pi_{xy} - \pi_{xy}^s \right)^2,
\end{aligned}
\end{equation}
where $\eta$ is a positive constant controlling the convergence. We have the following result on the distributed algorithm.
\begin{proposition} \label{prop:iterations}
 The iterative steps of ADMM to solve (OP-B) are summarized as follows:
\begin{equation}\label{ADMM1_eqn1}
\begin{split}
    \Pi_{x}^t(k+1) &\in \arg \min_{\Pi_x^t\in\mathcal{F}_{x}^t}  U_x w(p_x(\Pi_x^t))\\ 
    &+\sum_{y\in\mathcal{Y}_x} \alpha_{xy}^t(k) \pi_{xy}^t + \frac{\eta}{2} \sum_{y\in\mathcal{Y}_x} (\pi_{xy}^t - \pi_{xy}(k))^2,
\end{split}
\end{equation}
\begin{equation}\label{ADMM1_eqn2}
\begin{aligned}
        \Pi_{y}^s(k+1) &\in \arg \min_{\Pi_{y}^s\in\mathcal{F}_{y}^s} - \sum_{x\in\mathcal{X}_y} \tau_y s_{xy}(\pi_{xy}^s) \\ &-\sum_{x\in\mathcal{X}_y} \alpha_{xy}^s(k)\pi_{xy}^s + \frac{\eta}{2} \sum_{x\in\mathcal{X}_y} (\pi_{xy}(k) - \pi_{xy}^s)^2,
\end{aligned}
\end{equation}
\begin{equation}\label{ADMM1_eqn3}
\begin{split}
    \pi_{xy}(&k+1)= \arg \min_{\pi_{xy}} - \alpha_{xy}^t(k)\pi_{xy} + \alpha_{xy}^s(k)\pi_{xy} \\
    &+\frac{\eta}{2}(\pi_{xy}^t(k+1) - \pi_{xy})^2 + \frac{\eta}{2}(\pi_{xy} - \pi_{xy}^s(k+1))^2,
\end{split}
\end{equation}
\begin{equation}\label{ADMM1_eqn4}
\begin{split}
    \alpha_{xy}^t(k+1) = \alpha_{xy}^t(k) + \eta(\pi_{xy}^t(k+1)-\pi_{xy}(k+1))^2,
\end{split}
\end{equation}
\begin{equation}\label{ADMM1_eqn5}
\begin{split}
    \alpha_{xy}^s(k+1) = \alpha_{xy}^s(k) + \eta(\pi_{xy}(k+1)-\pi_{xy}^s(k+1))^2,
\end{split}
\end{equation}
where $\Pi_{\tilde{x}}^t:=\{\pi_{xy}^t\}_{y\in\mathcal{Y}_x,x=\tilde{x}}$ represents the solution at target node $\tilde{x}\in\mathcal{X}$, and $\Pi_{\tilde{y}}^s:=\{\pi_{xy}^s\}_{x\in\mathcal{X}_y,y=\tilde{y}}$ represents the proposed solution at source node $\tilde{y}\in\mathcal{Y}$. In addition, $\mathcal{F}_{x}^t := \{ \Pi_{x}^t | \pi_{xy}^t \geq 0, y\in\mathcal{Y}_x, \underline{p}_x \leq \sum_{y \in \mathcal{Y}_x} \pi_{xy}^t \leq \bar{p}_x\}$, and $\mathcal{F}_{y}^s := \{ \Pi_{y}^s | \pi_{xy}^s \geq 0, x\in\mathcal{X}_y, \underline{q}_y \leq \sum_{x \in \mathcal{X}_y} \pi_{xy}^s \leq \bar{q}_y\}$.
\end{proposition}
\begin{proof}
The proof follows similarly to the proof of Proposition 1 in \cite{jhughes2021fair}.
\end{proof}

The iterations in \eqref{ADMM1_eqn1}-\eqref{ADMM1_eqn5} can be further simplified to four iterations.
\begin{proposition}\label{prop:simplified_iterations}
The iterations \eqref{ADMM1_eqn1}-\eqref{ADMM1_eqn5} can be simplified as follows:
\begin{equation}\label{ADMM2_eqn1}
\begin{split}
    \Pi_{x}^t(k+1) &\in \arg \min_{\Pi_x^t\in\mathcal{F}_{x}^t}  U_x w(p_x(\Pi_x^t))\\ 
    &+\sum_{y\in\mathcal{Y}_x} \alpha_{xy}(k) \pi_{xy}^t + \frac{\eta}{2} \sum_{y\in\mathcal{Y}_x} (\pi_{xy}^t - \pi_{xy}(k))^2,
\end{split}
\end{equation}
\begin{equation}\label{ADMM2_eqn2}
\begin{aligned}
        \Pi_{y}^s(k+1) &\in \arg \min_{\Pi_{y}^s\in\mathcal{F}_{y}^s} - \sum_{x\in\mathcal{X}_y} \tau_y s_{xy}(\pi_{xy}^s) \\ &-\sum_{x\in\mathcal{X}_y} \alpha_{xy}(k)\pi_{xy}^s + \frac{\eta}{2} \sum_{x\in\mathcal{X}_y} (\pi_{xy}(k) - \pi_{xy}^s)^2,
\end{aligned}
\end{equation}
\begin{equation}\label{ADMM2_eqn3}
\begin{split}
    \pi_{xy}(k+1) = \frac{1}{2} \left(\pi_{xy}^{t}(k+1) + \pi_{xy}^{s}(k+1)\right),
\end{split}
\end{equation}
\begin{equation}\label{ADMM2_eqn4}
\begin{split}
    \alpha_{xy}(k+1) = \alpha_{xy}(k) + \frac{\eta}{2}\left(\pi_{xy}^{t}(k+1) - \pi_{xy}^{s}(k+1)\right).
\end{split}
\end{equation}
\end{proposition}
\begin{proof}
The proof follows similarly to the proof of Proposition 2 in \cite{jhughes2021fair}.
\end{proof}

We summarize the results into the following Algorithm \ref{Alg:1}.

\begin{algorithm}[h]
\caption{Distributed Algorithm}\label{Alg:1}
\begin{algorithmic}[1]
\While {$\Pi_{x}^t$ and $\Pi_{y}^s$ not converging}
\State Compute $\Pi_{x}^t(k+1)$  using \eqref{ADMM2_eqn1}, for all $x\in\mathcal{X}_y$
\State Compute $\Pi_{y}^s(k+1)$  using \eqref{ADMM2_eqn2}, for all $y\in\mathcal{Y}_x$
\State Compute $\pi_{xy}(k+1)$  using \eqref{ADMM2_eqn3}, for all $\{x,y\}\in \mathcal{E}$
\State Compute $\alpha_{xy}(k+1)$  using \eqref{ADMM2_eqn4}, for all $\{x,y\}\in \mathcal{E}$
\EndWhile
\State \textbf{return} $\pi_{xy}(k+1)$, for all $\{x,y\}\in \mathcal{E}$
\end{algorithmic}
\end{algorithm}

\textit{\textbf{Remark:}} The developed algorithm can be interpreted as a negotiation process between each pair of connected resource owner and target node on the security resource allocation scheme, which does not require a central planner. The final outcome indicates that the negotiation reaches a consensus.

\section{Case Studies}
In this section, we corroborate the developed results by focusing on the impacts of the behavioral consideration on the security investment decision-making. We investigate a transport network consisting of two source nodes (security resource owners) and five target nodes (assets to protect). We define the loss parameter at each target as $U_1 = 12$, $U_2 = 9$, $U_3= 5$, $U_4 = 3$, and $U_5=2$. Additionally, we use Prelec's probability weighting function defined in \eqref{eqn:weighting_func} and the probability function shown in \eqref{eqn:prelec}. We set the upper bound of security resources to $\bar{q}_1 = 10$ units and $\bar{q}_2=4$ units, meaning the maximum amount of security resources the sources can invest. The utility functions of $s_{xy}$ are considered to be linear.  

\subsection{Impact of Behavioral Consideration}
We first examine how incorporating the behavioral considerations will impact the transportation of security resources from the sources to targets. This involves looking at how the parameter $\gamma$ affects the outcome of the resource allocation. 
We expect that as the parameter $\gamma$ goes to one, the amount of resources received at each target should be the same as when misperception is not considered, i.e., the objective function would be $U_x(p_x(\Pi_x))$. Fig. \ref{fig:gamma_a} corroborates this result. We can also observe that the target node with a larger value $U_x$ receives more resources, indicating that the system planner prefers to secure more valuable targets under a constrained budget. The relationship between the amount of received resources at targets follows from the order of node's value $U_x$ in Assumption \ref{assump:2}.
Additionally, as the behavioral parameter $\gamma$ goes to 1, the aggregated loss at the targets converges to the same value when misperception is not considered, as shown in Fig. \ref{fig:gamma_b}. It can also be seen that the bounded rational security investment strategy is not as efficient as the one under the accurate perception.

\begin{figure}[!t] 
\centering
\subfigure[]{
\includegraphics[width=0.7\columnwidth]{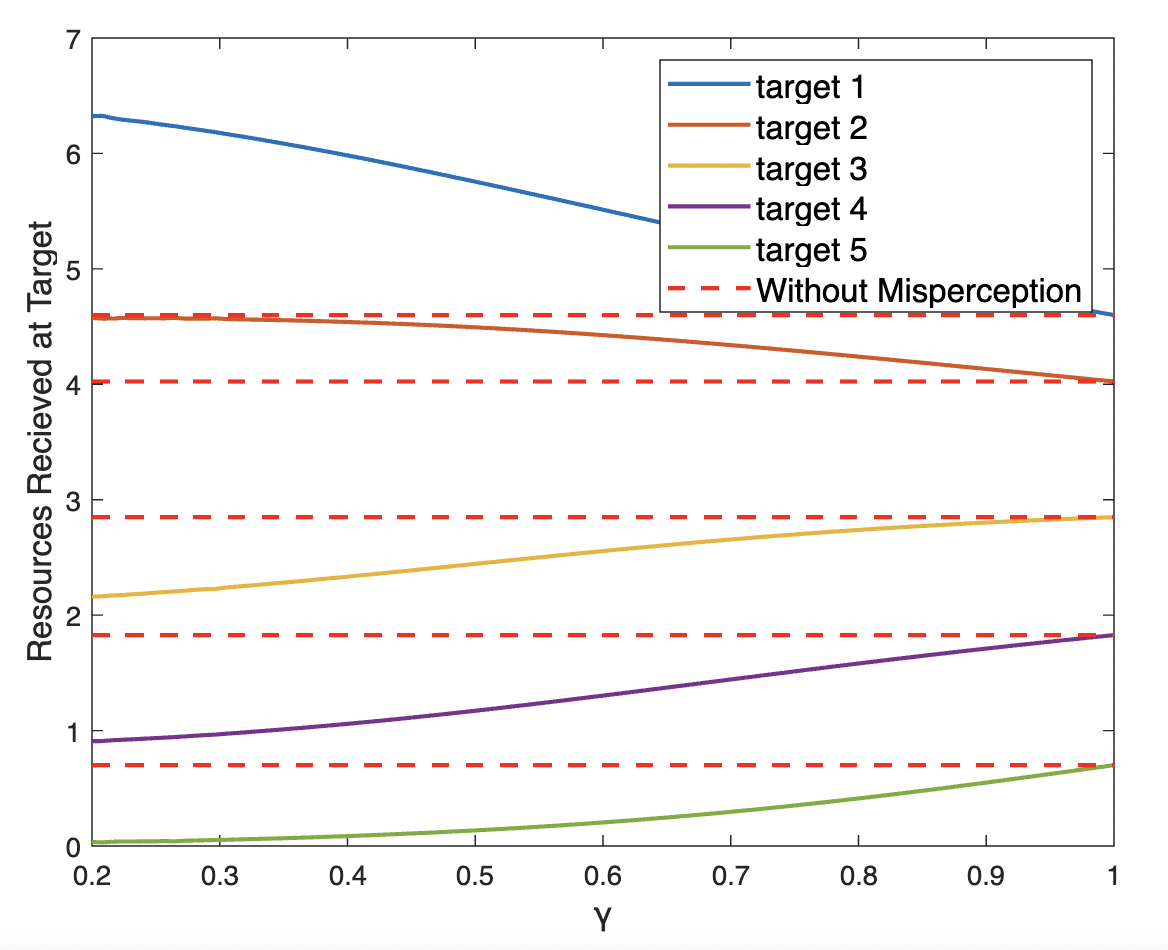}\label{fig:gamma_a}}
\subfigure[]{
\includegraphics[width=0.7\columnwidth]{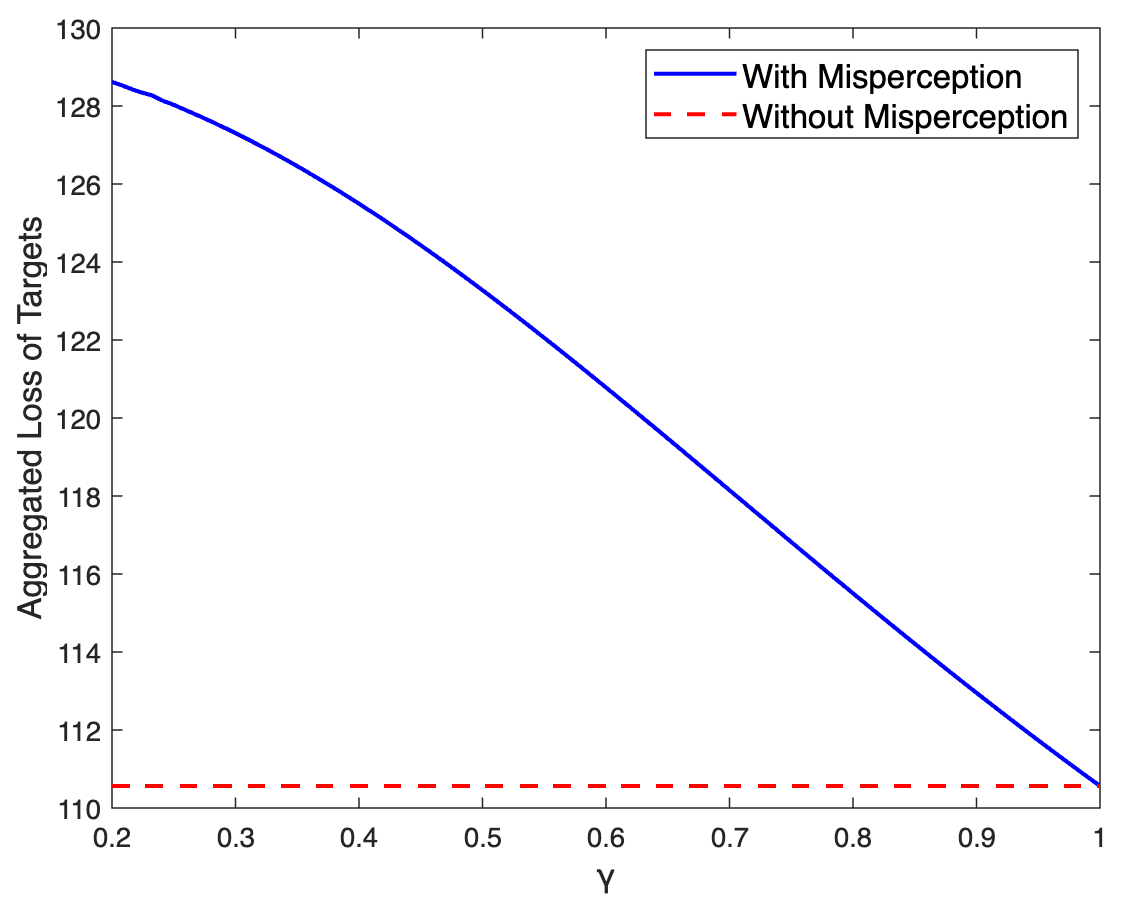}\label{fig:gamma_b}}
\caption{Impact of behavioral misperception on the security resource allocation plan. (a): Impact of $\gamma$ on the amount of resources received at each target node. The solution converges to the one without misperception as $\gamma$ goes to 1. (b): Aggregated loss of the targets with varying $\gamma$. A larger degree of misperception yields a less efficient resource transport strategy.}\vspace{-5mm}
\label{fig:gamma}
\end{figure}

\subsection{Performance of Distributed Algorithm}
We next show the performance of the proposed distributed algorithm in Algorithm \ref{Alg:1}. We use Algorithm \ref{Alg:1} to solve the optimization problem in (OP-B). Fig. \ref{fig:dist_a} shows that the distributed algorithm can efficiently converge to the centralized optimal solution. We also examine how the parameter $\tau_y$ influences the outcome of the transport plan. In the case study, $\tau_y$ is set to be the same at every source node, i.e., $\tau_y = \tau,\ \forall y\in\mathcal{Y}$. We leverage the developed distributed algorithm to compute the optimal strategy for various $\tau\in[0,1]$, and Fig. \ref{fig:dist_b} depicts the results. We can observe that when $\tau$ goes to zero in (OP-B), the aggregated loss of target nodes under the obtained strategy coincides with the one to (OP-A). This result makes sense as the utility term $s_{xy}$ no longer plays a role in (OP-B) when $\tau=0$. Another remark is that when $\tau<0.5$, the loss of targets under the solution to (OP-B) is smaller than the (OP-A) counterpart. This phenomenon is due to the system planner pays more attention to minimize the risks of targets when $\tau$ is small. As $\tau$ increases, the system planner cares more about maximizing the utility of resource owners, and thus it yields a larger loss of targets as shown in Fig. \ref{fig:dist_b}.

\begin{figure}[!t] 
\centering
\subfigure[]{\includegraphics[width=0.7\columnwidth]{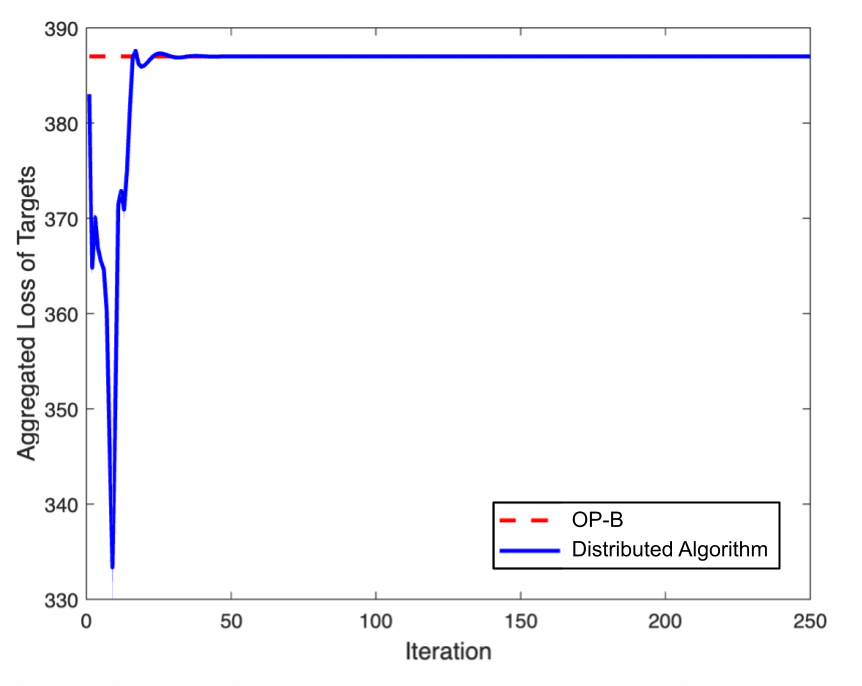}\label{fig:dist_a}}
\subfigure[]{\includegraphics[width=0.7\columnwidth]{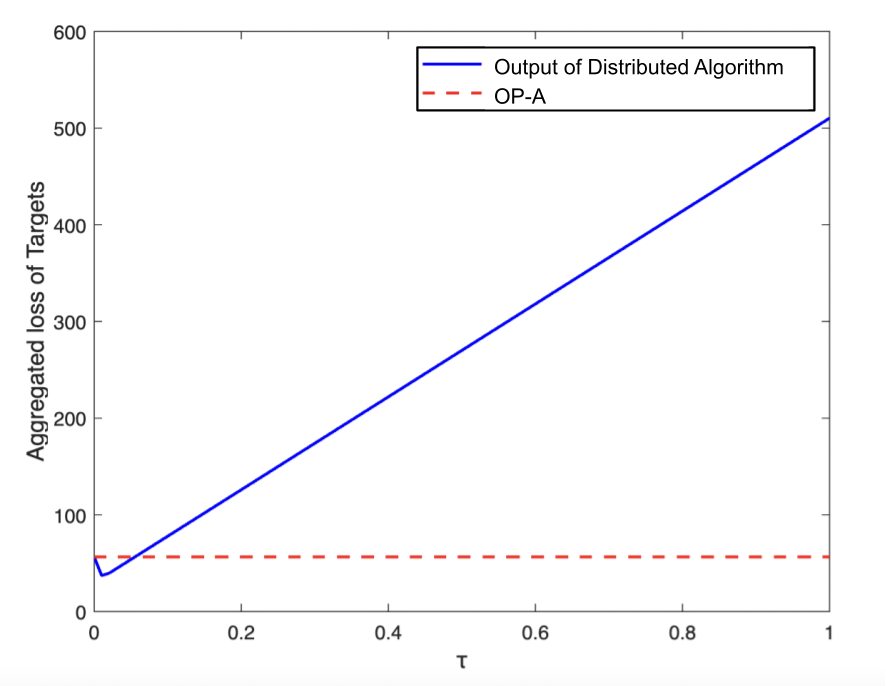}\label{fig:dist_b}}
\caption{(a): Effectiveness of the distributed algorithm \ref{Alg:1}. The distributed algorithm converges to the centralized optimal solution to problem (OP-B). (b): Impact of weighting constant $\tau$ on the aggregated loss of the targets. The solution degenerates to the one to problem (OP-A) as $\tau$ goes to 0.}\vspace{-5mm}
\label{fig:dist}
\end{figure}

\section{Conclusion}
This paper has developed a behavioral framework for security investments over a network consisting of multiple source nodes and heterogeneous target nodes. The behavioral element captures human's misperception of the successful attack probabilities at targets under a given level of security investment. The analysis has shown that the bounded-rational optimal resource allocation admits a sequential water-filling nature. In addition, we have discovered that fewer targets will receive security resources under the behavioral paradigm compared with the non-behavioral setting, revealing the sub-optimal feature of the strategy due to the behavioral misperception. We have further developed an efficient distributed algorithm to compute the resource allocation plan with a convergence guarantee which enjoys advantages when the transport network becomes enormous and complex. The case studies have corroborated that the system planner favors the higher-valued targets in bounded rational security investment, often resulting in lower-valued targets receiving smaller amounts or no resources.
Future works include extending the current framework to an adversarial setting and develop resilient security investment strategies. Another direction is to investigate dynamic resource allocation over a time-horizon with bounded rational agents.

\bibliographystyle{IEEEtran}
\bibliography{references.bib}

\appendix
\subsection{Proof of Lemma \ref{lemma:convex}}\label{lemma_proof:convex}
Here, we show that the second derivative of $w(p_x(\Pi_x))$ with respect to $\pi_{xy}$, $\forall y\in\mathcal{Y}_x$, is positive.
Using the probability function in \eqref{eqn:weighting_func}, we have
\begin{equation*}
 \begin{aligned}
 \frac{\partial^2 w(p_x(\Pi_x)) }{\partial \pi_{xy}^2} = -\gamma(\gamma-1)(-\log(p_x(\Pi_x)))^{\gamma-2}  w(p_x(\Pi_x)) \\
 \cdot \left(\frac{\partial p_x(\Pi_x)}{\partial \pi_{xy}}\Big /p_x(\Pi_x)\right)^2 + \gamma(-\log(p_x(\Pi_x)))^{\gamma-1} w(p_x(\Pi_x)) \\
 \cdot \frac{p_x(\Pi_x) \cdot \frac{\partial^2 p_x(\Pi_x)}{\partial \pi_{xy}^2} - \big(\frac{\partial p_x(\Pi_x)}{\partial \pi_{xy}}\big)^2}{(p_x(\Pi_x))^2}+ \Big(\gamma (-\log(p_x(\Pi_x)))^{\gamma-1} \\
 \cdot \frac{\partial p_x(\Pi_x)}{\partial \pi_{xy}}\Big/p_x(\Pi_x)\Big)^2 w(p_x(\Pi_x)).
 \end{aligned}
 \end{equation*}
 The first and third terms multiply out to be positive, because of Assumption \ref{assump:1}. The second term may or may not be positive depending on $(p_x(\Pi_x) \cdot \frac{\partial^2 p_x(\Pi_x)}{\partial \pi_{xy}^2} - \big(\frac{\partial p_x(\Pi_x)}{\partial \pi_{xy}}\big)^2)/(p_x(\Pi_x))^2$. If the second term term is positive than the second derivative is positive and thus the function is convex. If the term is negative we need to show that: 
 \begin{equation*}
 \begin{aligned}
 \Big[-\gamma(\gamma-1)(-\log(p_x(\Pi_x)))^{\gamma-2} \left(\frac{\partial p_x(\Pi_x)}{\partial \pi_{xy}}\Big/p_x(\Pi_x)\right)^2\\
+ \Big(\gamma (-\log(p_x(\Pi_x)))^{\gamma-1} \cdot \frac{\partial p_x(\Pi_x)}{\partial \pi_{xy}}\Big/p_x(\Pi_x)\Big)^2\Big]  \\
 \cdot w(p_x(\Pi_x)) > 
 \gamma(-\log(p_x(\Pi_x)))^{\gamma-1} w(p_x(\Pi_x))\\
 \cdot \frac{p_x(\Pi_x) \cdot \frac{\partial^2 p_x(\Pi_x)}{\partial \pi_{xy}^2} - \big(\frac{\partial p_x(\Pi_x)}{\partial \pi_{xy}}\big)^2}{(p_x(\Pi_x))^2} .
 \end{aligned}
 \end{equation*}
With cancellation and some algebraic manipulation, it is easy to see that the statement is true. Thus, the objective function is strictly convex.

\end{document}